\documentclass{eptcs}
\providecommand{\event}{CompMod 2009} 
\providecommand{\volume}{??}
\providecommand{\anno}{2009}
\providecommand{\firstpage}{1}
\providecommand{\eid}{??}
\usepackage{breakurl}        

\usepackage{amsmath}
\usepackage{amssymb}
\usepackage{dsfont}
\usepackage{theorem}
\usepackage{longtable}
\usepackage{pictexwd}
\usepackage{epsfig}
\usepackage{graphicx}

\newtheorem{thm}{Theorem}[section]
 \newtheorem{cor}[thm]{Corollary}
 \newtheorem{lem}[thm]{Lemma}
 \newtheorem{prop}[thm]{Proposition}
 \newtheorem{defn}[thm]{Definition}
\newcommand{\qed}{\hfill \mbox{$\Box$}}
\newenvironment{proof}{\vspace{1ex}\noindent{\it Proof}\hspace{0.5em}}
	{\hfill\qed\vspace{1.5ex}\newline}

\newcommand{\sgn}{\mbox{${\rm sgn}$}}

\newcommand{\haken}[1]{\mbox{$\{1,\dots,#1\}$}}             

\newcommand{\card}[1]{\mbox{${\rm card}\,#1$}}

\title{Dynamical and Structural Modularity of\\ Discrete Regulatory Networks}
\author{Heike Siebert
\institute{DFG Research Center {\sc Matheon}, Freie Universit\"at Berlin\\
Arnimallee 6\\
D-14195 Berlin, Germany}
\email{siebert@mi.fu-berlin.de}
}
\def\titlerunning{Modularity of Discrete Regulatory Networks}
\def\authorrunning{H. Siebert}
\begin{document}
\maketitle

\begin{abstract}
A biological regulatory network can be modeled as a discrete function $f$ that contains all available information on network component interactions. From $f$ we can derive a graph representation of the network structure as well as of the dynamics of the system. In this paper we introduce a method to identify modules of the network that allow us to construct the behavior of $f$ from the dynamics of the modules. Here, it proves useful to distinguish between dynamical and structural modules, and to define network modules combining aspects of both.
As a key concept we establish the notion of symbolic steady state, which basically represents a set of states where the behavior of $f$ is in some sense predictable, and which gives rise to suitable network modules.
We apply the method to a regulatory network involved in T helper cell differentiation.
\end{abstract}
\section{Introduction}
Qualitative methods present a rigorous mathematical framework for modeling biological systems for which experimental data needed to determine kinetic parameters and mechanisms is lacking. The components of the system are modeled as variables adopting only finitely many values, so-called activity levels. In the simplest case, we obtain a Boolean representation, where the values 0 and 1 may for example represent a gene being inactive or active. 
In the general case, each component can have several activity levels, which may be appropriate depending on the biological data, and often is useful when modeling components that influence several other network components. A vector assigning each component an activity level then represents a state of the system. The information about network structure as well as the logical rules governing the behavior of the system in state space is represented by a discrete function $f$.

Although discrete networks are a strongly simplified representation of the original system, complex networks are hard to analyze, not least because the state space grows exponentially with the number of components.
So, methods to reduce the complexity of the analysis are of great interest. One approach is to deconstruct the network in smaller building blocks that can be analyzed more easily, which leads to the notion of network modularity.

The idea of decomposing networks into modules is well-established in systems biology, although the notion of network module is not clear-cut. Often modules are defined based on biological criteria, that have to be translated into mathematical properties in order to identify them in a mathematical model (see e.\,g. \cite{Hartwell,Ederer,Papin}). Other approaches focus purely on the the graph representation of the network structure. Modules are defined as subgraphs satisfying graph theoretical characteristics often related to connectivity \cite{Albert2}, or with statistical significance in comparison with random networks \cite{Newman,Alon}. In addition to this structural view, there are also approaches to find dynamical modules, see e.\,g.\ \cite{Irons}, that focus on identifying behavioral characteristics. However, in general the results obtained by analyzing such modules in isolation do not translate to the original network, since additional influences have to be taken into account once the module is re-embedded in the original system.
Here, the key is finding conditions that allow to draw conclusions about a complex network from knowledge obtained from module analysis, as e.\,g. possible in the modular response analysis approach in the context of metabolic networks and steady state fluxes \cite{Kholodenko, Bruggeman}.

In this paper, we focus on the discrete modeling approach, presenting a method to identify \emph{network modules} that allow us to derive precise information on the dynamics of the original system from the results of the analysis of the modules, building on ideas and significantly extending results from \cite{MCS,Siebert09-AB}. In particular, we show that we can explicitly construct attractors of the original systems from network module attractors. Here, modularity is a key concept, and we exploit a purely structural as well as a purely dynamical view of modularity to eventually determine network modules combining important aspects of both. The core notion in our method is that of \emph{symbolic steady state}. Such a state represents a set of constraints on the activity levels of the network components that allows us on the one hand to focus on dynamics restricted to subsets of state space, on the other hand enables us to identify dynamical and structural modules that render the basis for defining suitable network modules.

The paper is organized as follows. In the next section we describe the discrete modeling formalism we use throughout the paper, and introduce structural, dynamical and network modules. In Sect.~\ref{sss} we establish the notion of symbolic steady state as well as related concepts. This is followed by the main results concerning network analysis utilizing modules in Sect.~\ref{ana}. We then illustrate the results for a class of networks, namely networks with input layer. In Sect.~\ref{Tcellsect} we apply the method to the analysis of a regulator network involved in T helper cell differentiation proposed in \cite{Mendoza}. We close with conclusions and perspectives.
\section{Discrete regulatory networks}\label{prelim}
In this paper we model regulatory systems as discrete functions which capture all available information about network interactions and the logical rules governing the behavior of the system. Throughout the text, let us consider a system consisting of $n\in\mathbb{N}$ network components $\alpha_1,\dots,\alpha_n$. In the following we identify a component $\alpha_i$ with its index $i$ to simplify notation. Each component is interpreted as a variable which takes integer values that represent the different activity levels of the component. The literal meaning of those levels may be very different for different network components, for example they can represent levels of substance concentration, gene activity, presence or absence of a signal and so on. 
A vector assigning each component an activity level represents a state of the system, and the dynamics of the system is represented by state changes due to component interactions.
\begin{defn}\label{nw}
For all $i\in\haken{n}$, let $p_i\in\mathbb{N}$, and set $X_i=\{0,1,\dots,p_i\}$. Set $X=X_1\times\dots\times X_n$, and let $f=(f_1,\dots,f_n):X\to X$ be a function.
We call $f$ a \emph{network} comprising $n$ components. For each $i\in\haken{n}$, the value $p_i$ is the \emph{maximal activity level} of $\alpha_i$, and $X_i$ is called the \emph{range} of $\alpha_i$. The set $X$ is called the \emph{state space} of $f$.
\end{defn}
Each coordinate function $f_i$ of $f$ describes the rules governing the behavior of the $i$-th network component depending on the state of the system. But $f$ carries not only dynamical but also structural information on the system. Both aspects can be represented by directed graphs derived from $f$ as we will see in the following two subsections. In the remainder of the paper $f$ denotes a network as introduced in Def.~\ref{nw}
\subsection{Structure}
We represent the structure of a network by a signed directed (multi-)graph, where vertices represent the network components, and an edge from $\alpha_i$ to $\alpha_j$ signifies that the value of $f_j$ depends on the activity level of $\alpha_i$. The sign of the edge represents the character, i.\,e., activating or inhibiting, of the interaction. This description is inherently local in nature, so we first introduce a structural representation depending on the state of the system. This notion was introduced for Boolean functions in \cite{RRT} and is used for multi-value functions in the form considered here in \cite{RichardArxiv}.
\begin{defn}\label{localig}
Let $x\in X$. By $G(f)(x)$ we denote the directed signed (multi-)graph with vertex set $V=\{\alpha_1,\dots,\alpha_n\}$ and edge set $E(x)\subseteq V\times V\times\{+,-\}$. An edge $(i,j,\varepsilon)$ belongs to $E(x)$ iff there exists $c_i\in\{-1,+1\}$ such that $x_i+c_i\in X_i$ and
\[
\sgn\,\,\left(\frac{f_j((x_1,\dots,x_{i-1},x_i+c_i,x_{i+1},\dots,x_n))-f_j(x)}{c_i}\right)=\varepsilon\,,
\]
where the function $\sgn:\mathbb{Z}\to\{+,-,0\}$ satisfies $\sgn(0)=0$, $\sgn(z)=\,+$\, if $z>0$, and $\sgn(z)=-$\, if $z<0$.
We call $G(f)(x)$ the \emph{local interaction graph of $f$ in $x$}.
\end{defn}
The local interaction graph in $x$ is closely related to the discrete Jacobian matrix as introduced in \cite{Robert} in the Boolean case. Note that in the multi-value other than in the Boolean case it is possible that $G(f)(x)$ contains parallel edges. There are at most two parallel edges from one vertex to another which then have opposite sign.

The local definition is easily extended, if we are interested in a representation of the interactions influencing the system behavior in larger subsets of state space.
\begin{defn}\label{IG}
Let $Y\subseteq X$. We denote by $G(f)(Y)$ the union of the graphs $G(f)(x)$, $x\in Y$. We denote the graph $G(f)(X)$ also by $G(f)$ and call it the \emph{global interaction graph of $f$}.
\end{defn}
In Fig.~\ref{Ex} (b) we see the global interaction graph of the network defined in Fig.~\ref{Ex} (a). The local interaction graph in the state $(1,1,0)$ is shown in Fig.~\ref{symb} (a). To simplify notation we often write $G(x)$ and $G(Y)$ instead of $G(f)(x)$ and $G(f)(Y)$, respectively, if the corresponding function $f$ is clear from the context.

When analyzing interaction graphs we are in particular interested in modules of the graph, a term for which there exists a variety of definitions as mentioned in the introduction. For our purposes it is convenient to use the term in the broadest sense, initially.
\begin{defn}
A directed (multi-)graph $G'=(V^{G'}, E^{G'})$ is called a subgraph of a directed (multi-)graph $G=(V^G,E^G)$ if $V^{G'}\subseteq V^G$, $E^{G'}\subseteq E^G$, and every edge in $E^{G'}$ has both end-vertices in $V^{G'}$.
We call a subgraph of $G(f)$ \emph{structural module of $f$}.
\end{defn}
Note that for a structural module the vertex set, the edge set or both may be smaller than for $G(f)$. For example, the graph shown in Fig.~\ref{Ex} (d) is a structural module of the function $f$ given in the same figure. In general, local interaction graphs are structural modules of $f$. Of course, this definition is not very useful for analyzing the network structure or finding characteristics of the system. However, it is of conceptual advantage for us in the endeavor of defining network modules that combine structural and dynamical characteristics.
\begin{figure}[tb]
\begin{center}
\includegraphics[width=.9\textwidth]{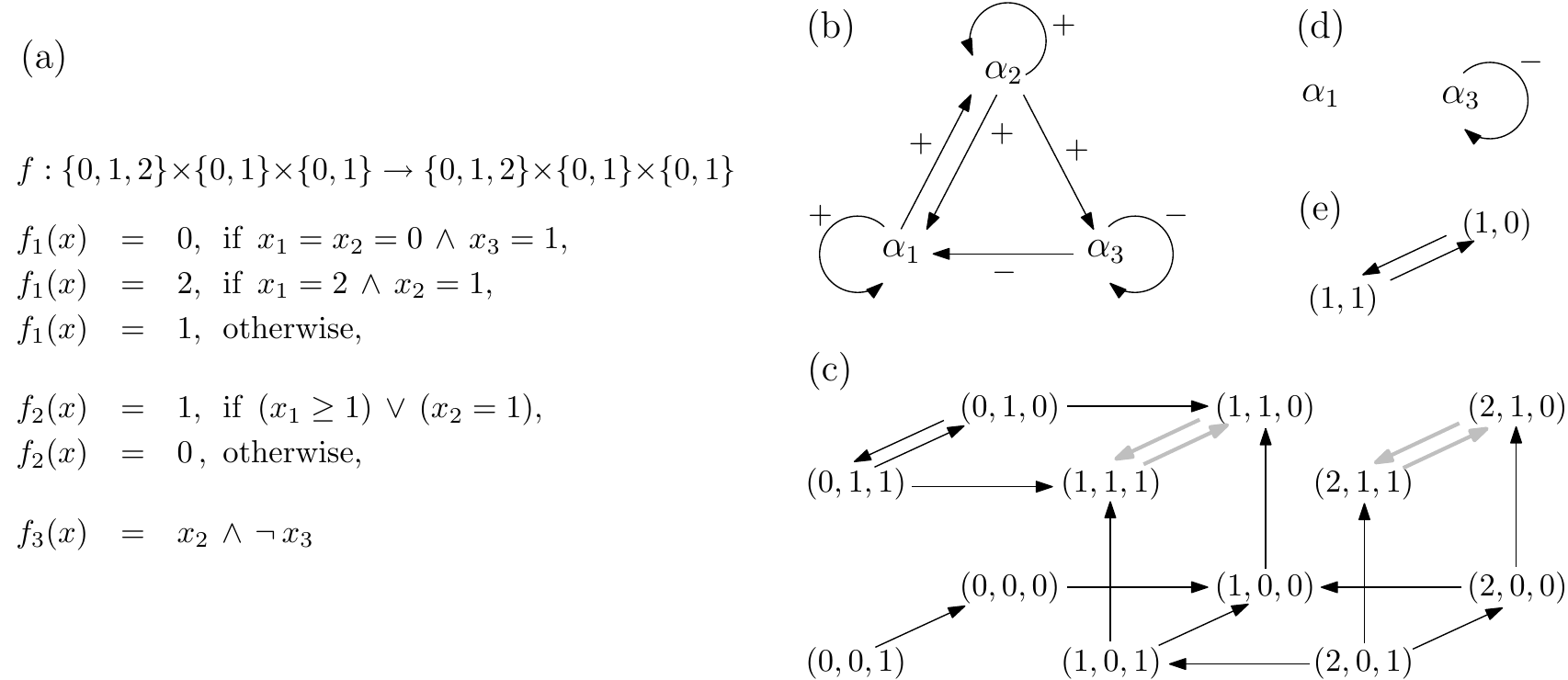}
\caption{\label{Ex} A network (a), its global interaction graph (b), and its state transition graph (c), where attractors are indicated by fat gray edges. Logical disjunction, conjunction and negation are represented by $\vee$, $\wedge$, and $\neg$\,, respectively. In (d) a structural, in (d) a dynamical module of $f$.}
\end{center}
\end{figure}
\subsection{Dynamics}
There are different approaches to deriving the dynamics of $f$. Commonly used is the so-called synchronous update strategy, where the successor of a state $x$ is its image under $f$. A more realistic assumption is that not all changes indicated by differences in component values of $x$ and $f(x)$ take the same amount of time to be executed, since they may represent very different biological processes. However, we lack the information to decide which of those processes of activity level change is the fastest. Therefore, all possible state transitions are taken into account leading to a non-deterministic representation of the dynamical behavior. Furthermore, we assume that a component value changes only by absolute value one in each transition, even though the function value may indicate a bigger change. This update method is called \emph{asynchronous update} \cite{ThomasDAri,ThomasChaos}.
\begin{defn}\label{STG}
We denote by $S(f)$ the directed graph with vertex set $X$ and edge set $E(S(f))$ defined as follows. An edge $(x,x')$ is in  $E(S(f))$ for states $x=(x_1,\dots,x_n),x'=(x'_1,\dots,x'_n)\in X$ if and only if $x'=f(x)=x$ or $x'_i=x_i+\sgn(f_i(x)-x_i)$ for some $i\in\haken{n}$ satisfying $x_i\neq f_i(x)$, and $x_j'=x_j$ for all $j\neq i$. We call $S(f)$ the \emph{asynchronous state transition graph of $f$}.
\end{defn}
To analyze state transition graphs we use, in addition to standard terminology from graph theory such as paths and cycles, the following concepts.
\begin{defn}
An infinite path $(x^0,x^1,\dots)$ in $S(f)$ is called \emph{trajectory}.
A nonempty set of states $D$ is called \emph{trap set} if every trajectory starting in $D$ never leaves $D$.
A trap set $A$ is called \emph{attractor} if for all $x^1,x^2\in A$ there is a path from $x^1$ to $x^2$ in $S(f)$.  Attractors of cardinality greater than one are called \emph{cyclic} attractors. A state $x$ is called \emph{steady state}, if there exists an edge $x\to x$, i.\,e. if $f(x)=x$.
\end{defn}
It is easy to see that each trap set contains at least one attractor, and that
attractors are the terminal strongly connected components of $S(f)$. They represent asymptotically stable behavior and often have clear biological meaning.

In Fig.~\ref{Ex} (c) we see the state transition graph for the network introduced in (a). The system has two cyclic attractors, namely $\{(1,1,0),(1,1,1)\}$ and $\{(2,1,0),(2,1,1)\}$.

As for the structural graph, we can define modules of the state transition graph as subgraphs, i.\,e. sets of states and corresponding state transitions representing fractions of the system's dynamics. However, it may also be of interest to only focus on the behavior of a subset of network components, which we can derive from the state transition graph by projection.
\begin{defn}\label{dynmod}
Let $S'=(Y,E(S'))$ be a subgraph of $S(f)$. Let $\pi^I:X\to\prod_{i\in I}X_i$ be the projection on the components in the ordered set $I\subseteq\haken{n}$. We define $\pi^I(S')$ as the graph with vertex set $\pi^I(Y)$ and edges $\pi^I(v^1)\to\pi^I(v^2)$ for $v^1,v^2\in Y$ such that there exists an edge $v^1\to v^2$ in $S(f)$, and $v^1=v^2$ or $v^2_i=v^1_i+\sgn(f_i(v^1)-v^1_i)$ for some $i\in I$.
We call $\pi^I(S')$ \emph{dynamical module of $f$}.
\end{defn}
Any subgraph of $S(f)$ is a dynamical module of $f$, whith $I$ in the above definition chosen as the set $\{1,\dots,n\}$. Fig.~\ref{Ex} (e) shows a dynamical module of the function $f$ given in the same figure.
Here, we choose the subgraph $S'$ consisting of the cyclic attractor $\{(1,1,0),(1,1,1)\}$ and the corresponding edges between the two attractor states. Then the dynamical module $\pi^{\{1,3\}}(S')$ of $f$ is the graph in Fig.~\ref{Ex} (e).

Again, we do not incorporate any restrictions in the definition that ensure a significance of the modules, as e.\,g.\ in the approach in \cite{Irons} where the authors focus on projected dynamics that are independent of the behavior of the rest of the system.
\subsection{Network modules}
As mentioned in the introduction, we are looking for subnetworks of $f$ that are on the one hand easier to analyze than $f$ itself, and on the other hand carry information of importance for understanding the original system. We define modules of the network $f$ as follows.
\begin{defn}\label{nwmoddef}
We call a function $g$ a \emph{network module of $f$}, if there exist
$Y\subseteq X$ and an ordered set $I\subseteq\haken{n}$ such that $g:\pi^I(Y)\to\pi^I(X)$ satisfies
\[
\pi^I\circ f|_{Y}=g\circ\pi^I|_{Y},
\] where $|_{Y}$ denotes the restriction of a function to the set $Y$.\newline
We call $g$ \emph{autonomous}, if there exist integer intervals $Z_i=\{a_i,a_i+1,\dots,b_i\}$, $a_i\leq b_i$,  for all $i\in\haken{k}$, $k=\card{I}$, such that $\pi^I(Y)=Z_1\times\dots\times Z_k$, and if $g(\pi^I(Y))\subseteq\pi^I(Y)$.
\end{defn}
Let us again illustrate the notion using the example introduced in Fig.~\ref{Ex}. For $Y=\{(1,1,0),(1,1,1)\}$, we have $f(Y)=\{(1,1,0),(1,1,1)\}$. If we set $I=\{1,3\}$ and $g:\{1\}\times\{0,1\}\to\{0,1,2\}\times\{0,1\}$ with $g((1,1))=(1,0)$ and $g((1,0))=(1,1)$, then $\pi^I(f(x))=g(\pi^I(x))$ for $x\in Y$. Since $g(\pi^I(Y))\subseteq\pi^I(Y)$ holds, $g$ is an autonomous network module of $f$. Note that the set $Y$ is a set on which the behavior of $f$ can in some sense be characterized by the behavior of the components in $I$. If we add e.\,g.\ the state $(1,0,0)$ to the set $Y$ in our example, then there is no function $g$ satisfying $\pi^I\circ f|_{Y}=g\circ\pi^I|_{Y}$, since $\pi^I(1,0,0)=\pi^I(1,1,0)$, but $\pi^I(f(1,0,0))=(1,0)\neq(1,1)=\pi^I(f(1,1,0))$. That is, we cannot distinguish the behavior of $f$ in states $(1,0,0)$ and $(1,1,0)$ if we only have information on the components in $I$.

In general, network modules represent rather local aspects of the network in the sense that they describe the influences acting on a subset of components in a set of states. However, the information inherent in a network module does not necessarily suffice for determining dynamics beyond a single transition step. The second condition for autonomous network modules $g$ allows to apply $g$ iteratively on states in $\pi^I(Y)$, while the first ensures that we can derive trajectories according to the asynchronous update rule in a projection of state space. Moreover, the first condition allows to apply Def.~\ref{localig} to $g$. Thus, for autonomous network modules $g$ we can determine an interaction graph $G(g)$ and a state transition graph $S(g)$. By abuse of notation we also denote $G(g)$ the graph derived from $G(g)$ by renaming the vertices $1,\dots,k$ of $G(g)$ with the indices in $I$ while preserving the order. This allows us to identify the interaction graph of $g$ with a subgraph of $G(f)$.
\begin{lem}\label{aut}
Let $g$ be an autonomous network module as introduced in Def.~\ref{nwmoddef}. Then $G(g)$ is a structural and $S(g)$ is a dynamical module of $f$.
\end{lem}
\begin{proof}
As already mentioned, we associate each $i\in I$ with $l^i\in\{1,\dots,k\}$ via an order-preserving mapping, and rename each vertex of $G(g)$ with indices in $I$ according to this mapping.
Obviously, the vertex set of $G(g)$ resp.\ $S(g)$ is a subset resp.\ a projection via $\pi^I$ of a subset of the vertex sets of $G(f)$ resp.\ $S(f)$. Let $i\in I$, and choose $l^i\in\haken{k}$ as above, i.\,e., $\pi^i$ maps the $i$-th component of a state $x\in X$ to the $l^i$-th component in $\pi^I(X)$. Then we have $f_i(y)=(\pi^I(f(y)))_{l^i}=g_{l^i}(\pi^I(y))$ for all $y\in Y$.  Application of this equation to the conditions defining edges in Def.~\ref{localig} and \ref{STG} easily renders that each edge in $G(g)$ is also an edge of $G(f)$, and that $S(g)=\pi^I(S')$, where $S'$ denotes the subgraph of $S(f)$ with vertex set $Y$ and edges $y^1\to y^2$ of $S(f)$ with $y^1,y^2\in Y$.
\end{proof}
For the network module $g$ as defined as illustration for Def.~\ref{nwmoddef} the state transition graph is shown in Fig.~\ref{Ex} (e). We rename the vertex set $\{1,2\}$ of $G(g)$ with $I=\{1,3\}$. The graph $G(g)$ then consists of the vertices $\alpha_1$ and $\alpha_3$ and a negative loop on $\alpha_3$ as shown in Fig.~\ref{Ex}(d).

In the following sections, we focus on developing a method to determine network modules useful in the analysis of $f$.
\section{Symbolic steady states and frozen components}\label{sss}
Often network components are involved in a number of specific tasks.
Thus, although a network component may have a large range, only subsets of the range may be of interest when focusing on specific network behavior. To exploit this observation, we introduce the following notation. Here, we call the set $[a_i,b_i]:=\{a_i,a_i+1,\dots,b_i-1,b_i\}\subseteq X_i$ a discrete interval, if $a_i\leq b_i$, with $[a_i,a_i]:=\{a_i\}$. In the following, we identify $\{a_i\}$ with $a_i$ for all $a_i\in X_i$, $i\in\haken{n}$, and call $a_i$ \emph{regular value}. We will use intervals of cardinality greater than one instead of regular component values, if we do not have enough information to determine the exact component value. Following the terminology in \cite{MCS,Siebert09-AB}, we call intervals $[a_i,b_i]$ with $a_i<b_i$ \emph{symbolic values}.

We now need to integrate symbolic values in the dynamical analysis. Here, we generalize ideas from \cite{MCS,Siebert09-AB}.
\begin{defn}
For every $i\in\haken{n}$ let $X_i^{\Box}$ denote the set $\{[a_i,b_i]\subseteq X_i\mid a_i\leq b_i\}$ of discrete intervals in the range $X_i$. Set $X^{\Box}=X_1^{\Box}\times\dots\times X_n^{\Box}$. We call elements in $X$ \emph{regular}, elements in $X^\Box\setminus X$ \emph{symbolic states}. By $J(M)$ we denote the set of all symbolic valued components of $M$ for $M\in X^\Box$. Define
\[
F:X^{\Box}\to X^{\Box},\,\,M\mapsto(F_1(M),\dots,F_n(M))\,\mbox{ with }\,\, F_i(M)=[\,\min_{x\in M}f_i(x),\max_{x\in M}f_i(x)\,]\,\mbox{ for all }\,\,i\in\haken{n}.
\]
We call a state $M\in X^\Box\setminus X$ satisfying $F(M)=M$ \emph{symbolic steady state}.
\end{defn}
The elements of $X^\Box$ are subsets of $X$. The functions $f$ and $F$ coincide on the set $X$ of regular states which we identify with the elements of $X^\Box$ of cardinality one.
In general, if a component function value $f_i(M)$ is regular, this means there is enough information inherent in $M$ to exactly specify its value, while a symbolic value represents the fact that we have not enough information to do so. However, it may be possible to at least derive some constraint for the function value represented by the interval boundaries. For our running example given in Fig.~\ref{Ex} the state $([1,2],1,[0,1])=([1,2],1,X_3)=F(([1,2],1,X_3))$ is a symbolic steady state, where we can determine $F_2(([1,2],1,[0,1]))=1$ exactly, obtain the constraint that the first component cannot have value 0, but have no information on the third component.

We are particularly interested in regular components of a symbolic state $M$ that remain fixed on all trajectories starting in $M$. 
Keeping in mind that we consider the asynchronous update strategy, we can find a superset of the set of states reachable from $M$ by the following procedure. We define $\tilde{M}^0:=M$ and $\tilde{M}^k_j:=[\,\min (\tilde{M}^{k-1}_j\cup F_j(\tilde{M}^{k-1})),\,\max (\tilde{M}^{k-1}_j\cup F_j(\tilde{M}^{k-1}))\,]$ for all $k\in\mathbb{N}$. Since the boundaries of the intervals $M_j^k$ decrease resp.\ increase monotonously and are bounded by 0 resp.\ the maximal activity level $p_j$, the sequence $(\tilde{M}^k)_{k\in\mathbb{N}_0}$ converges to a symbolic state $\widetilde{M}$ representing a superset of the set of from $M$ in $S(f)$ reachable states. In particular, no trajectory starting in $\widetilde{M}$ can leave $\widetilde{M}$. We call $\widetilde{M}$ \emph{extended forward orbit of $M$}.
The next definition is in reference to the notion of frozen cores in random Boolean networks introduced by S. Kaufman \cite{Kauffman}.
\begin{defn}
Let $i\in\haken{n}$. If $M$ is a symbolic state with regular component $M_i$ such that $\widetilde{M}_i=M_i$ for the extended forward orbit $\widetilde{M}$ of $M$, then we say that the $i$-th network component is a \emph{frozen component of $M$}, or a component \emph{frozen to value $M_i$}.
The set $I$ of all frozen components of a symbolic state $M$ is called \emph{frozen core} of $M$, and is denoted by $(I,M)$.
\end{defn}
If a component $j$ of a symbolic state $M$ has symbolic value $X_j$, then of course the $j$-th component of the extended forward orbit is also $X_j$. For the example network in Fig.~\ref{Ex} the symbolic state $M=(X_1,1,X_3)$ coincides with its extended forward orbit. Thus, the frozen core of $M$ is given by $(\{2\}, (X_1,1,X_3))$.

Clearly, the frozen core of a symbolic steady state coincides with its set of regular components. Moreover, we can use the frozen core of a symbolic state to obtain a symbolic steady state, as the next statement shows.
\begin{thm}\label{frozen}
Let $(I,M')$ be the frozen core of a symbolic state $M'\in X^\Box$. Set $M^0:=\widetilde{M'}$ and $M^k:=F(M^{k-1})$ for all $k\in\mathbb{N}$. Then $(M^k)_{k\in\mathbb{N}}$ converges to a regular or a symbolic steady state $M$. We call $M$ the (symbolic) steady state derived from $(I,M')$.
\end{thm}
\begin{proof}
If the sequence converges to a limit $M$, clearly $F(M)=M$ follows from the definition of the sequence. The state $M$ is a regular or a symbolic steady state depending on the cardinality of $M$. We show convergence of $(M^k)_{k\in\mathbb{N}}$ by proving via induction that the sequence is decreasing monotonously with respect to the subset relation. That is, we show $M^{l+1}\subseteq M^l$ for all $l\in\mathbb{N}$.

For $i\in I$, we have $M_i^1=M_i^0$ by the definition of frozen components. If $i\in J(M^0)=\haken{n}\setminus I$, we have $M_i^1\subseteq M_i^0$ by the definition of the extended forward orbit.

Now, let $l\in\mathbb{N}$ and assume $M^{k}\subseteq M^{k-1}$ for all $k\leq l$. Recall that $M^{l+1}_i=F(M^l)_i=[\,\min_{x\in M^l}f_i(x),$ $\max_{x\in M^l}f_i(x)]$ and $M^{l}_i=F(M^{l-1})_i=[\,\min_{x\in M^{l-1}}f_i(x),\max_{x\in M^{l-1}}f_i(x)]$ for all $i\in\haken{n}$. Since $M^l\subseteq M^{l-1}$, we have $\min_{x\in M^{l-1}}f_i(x)\leq \min_{x\in M^l}f_i(x)\leq \max_{x\in M^l}f_i(x)\leq \max_{x\in M^{l-1}}f_i(x)$, and thus $M_i^{l+1}\subseteq M_i^l$ for all $i\in\haken{n}$.
\end{proof}
As mentioned above, for our running example $(\{2\}, (X_1,1,X_3))$ is the frozen core of the state $(X_1,1,X_3)$, where $X_1=\{0,1,2\}$ and $X_3=\{0,1\}$. Since the state coincides with its extended forward orbit, we start the sequence with $M^0=(X_1,1,X_3)=([0,2],1,[0,1])$, and $M^1=([1,2],1,[0,1])=M$ is the symbolic steady state derived from $(\{2\}, (X_1,1,X_3))$.
\section{Network analysis using modules}\label{ana}
In the following we want to determine network modules, such that the results of the module analysis can be directly used to obtain information on the behavior of the original system. Clearly, identification of such modules is generally only possible if we exploit at least some coarse knowledge of structural as well as dynamical characteristics of the original system. It turns out that the information inherent in a symbolic steady state $M$ is sufficient to determine network modules. We proceed by first associating a dynamical module with $M$, then derive a structural module, and finally define a network module suitable for utilization in network analysis.
\begin{figure}[tb]
\begin{center}
\includegraphics[width=\textwidth]{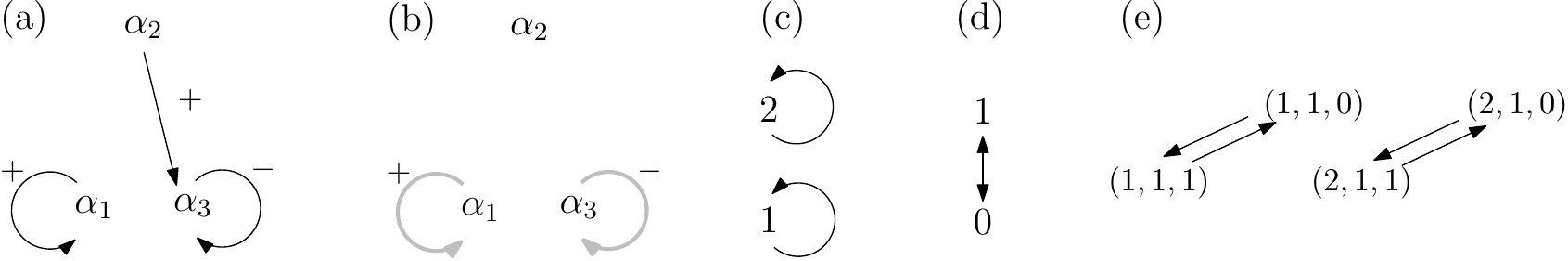}
\caption{\label{symb} Consider $f$ as given in Fig.~\ref{Ex}. In (a) the local interaction graph $G((1,1,0))$ of $f$. In (b) the graph $G(f|_M)$ for the symbolic steady state $M=([1,2],1,[1,0])$. Here, heavier gray edges indicate the two components of $G^\theta(M)$. In (c), (d) and (e) the state transition graphs $S(f^{Z_1})$, $S(f^{Z_2})$ and $S^M$, respectively.}
\end{center}
\end{figure}

Let us start by analyzing dynamical characteristics associated with a symbolic steady state $M$. If $x\in M$ is a regular state, then $f(x)\in M$ by definition of $F$. More precisely, $x_i+\sgn(f_i(x)-x_i)\in M_i$, since the interval bounded by $x_i$ and $f_i(x)$ is a subset of $M_i$. Thus, every trajectory starting in $M$ remains in $M$. We formulate this fact in the following statement.
\begin{prop}
If $M$ is a symbolic steady state, then the set of regular states represented by $M$ is a trap set.
\end{prop}
By definition, the subgraph of $S(f)$ with vertex set $M$ and all edges in $S(f)$ with both end-vertices in $M$ is a dynamical module. The result above shows that this module is of significance for the dynamical analysis of $f$, since a trap set always contains at least one attractor.

In order to associate a structural module with $M$, we have to recognize
some important properties attached to the regular components of $M$. While symbolic components may be dynamically active in the trap set $M$, i.\,e., the components change their values along at least some trajectories in $M$, the regular components remain fixed regardless of the behavior of the symbolic components in $M$. This means that the network components with symbolic values, i.\,e., components in $J(M)$, do not influence the behavior of the network components with regular values in the trap set $M$. In turn, the influence of the regular on the symbolic components remains the same for all states in $M$. This motivates the following definition describing the structural modules associated with $M$.
\begin{defn}
Let $M$ be a symbolic steady state. By $G^\theta(M)$ we denote the (multi-)graph with vertex set $V^\theta(M)=J(M)$ and edge set $E^\theta(M)=\{(i,j,\varepsilon)\in E^M \mid i,j\in J(M)\}$, where $E^M$ denotes the edge set of $G(f|_M)$.  We call a graph $Z=(V^Z,E^Z)$ \emph{component} of $G^\theta(M)$, if the undirected graph derived from $Z$ is a maximal connected subgraph of the undirected graph derived from $G^\theta(M)$.
\end{defn}
Note that we use the global interaction graph $G(f|_{M})$ instead of the local interaction graph $G(M)$ in the definition. The difference is that we only consider edges derived from component value changes in $M$ instead of in $X$ (compare Def.~\ref{localig} and \ref{IG}), and thus capture all interactions functional in $M$. In particular, there are no edges originating in frozen components of $M$, since it is not possible to vary their value without leaving $M$.

To illustrate the above notions let us again consider the example introduced in Fig.~\ref{Ex}. We have seen that the state $M=([1,2],1,[1,0])$ is a symbolic steady state. In Fig.~\ref{Ex} (c) we can easily see that the set of regular states $x\in M$ is a trap set which contains both attractors of the system. The global interaction graph $G(f|_M)$ is shown in Fig.~\ref{symb} (b). We obtain $G^\theta(M)$ simply by eliminating the vertex $\alpha_2$. The two components of $G^\theta(M)$ are the loops originating in $\alpha_1$ resp.\ $\alpha_3$.

The components of $G^\theta(M)$ are the structural modules we associate with $M$. In preparation for our definition of a network module derived from a structural module we need to verify that network components belonging to different components of $G^\theta(M)$ do not influence each others behavior in $M$. This property is captured in the following lemma, which has already been proved under slightly different conditions in \cite{MCS,Siebert09-AB}.
\begin{lem}\label{core}
Let $M$ be a symbolic steady state, and let $Z_1,\dots, Z_k$ be the components of $G^\theta(M)$. Consider a union $Z$ of arbitrary components $Z_j$. Let $x,y\in M$ such that $x_i=y_i$ for all $i\notin Z$. Then $f_i(x)=f_i(y)$ for all $i\notin Z$.
In particular, for $M'\in X^\Box$ such that $M'_i=M_i$ for all $i\notin Z$ and $M'_i\subseteq M_i$ for $i\in Z$, we have $F_i(M')=F_i(M)=M_i=M'_i$ for all $i\notin Z$.
\end{lem}
\begin{proof}
For $i\notin J(M)$, i.\,e., for a frozen component $i$ of $M$, we have $f_i(x)=f_i(y)=F_i(M)$ since $x,y\in M$.

Let $i\in J(M)\setminus Z$, and assume $f_i(x)\neq f_i(y)$. We know that if $x_j\neq y_j$ then $j\in Z$. Since $x,y\in M$, we can define a sequence $(x=x^1,x^2,\dots,x^m=y)$ in $M$ such that $x^l$ and $x^{l+1}$ differ in one component only, which is in $Z$, and the corresponding component values differ by absolute value one. Since $f_i(x)\neq f_i(y)$, it follows that $f_i(x^l)\neq f_i(x^{l+1})$ for some $l<m$. According to Def.~\ref{localig} there exists an edge in $G(f|_M)$ from some component in $Z$ to $i$, which is a contradiction.
\end{proof}
The lemma shows that the frozen core of $M$ constitutes a boundary between the components of $G^\theta(M)$ that enables us to analyze their behavior in isolation from each other.
To do so, we now derive a network module from a component $Z$ of $G^\theta(M)$ by defining a function $f^Z$.
\begin{lem}\label{nwmod}
Let $M$ be a symbolic steady state, and let $Z=(V^Z,E^Z)$ be a component of $G^\theta(M)$.
Set $k^Z={\rm card}\,V^Z$. Let $\iota^Z$ be an order preserving bijection from $\haken{k}$ to $V^Z$. Set $X^Z=M_{\iota^Z(1)}\times\dots\times M_{\iota^Z(k)}$. We define $f^Z:X^Z\to X^Z$, $f^Z=\pi^Z\circ F\circ\rho^Z$, where $\rho^Z:X^Z\to X^\Box$ with $\rho_i^Z(z)=M_i$ for $i\notin Z$ and $\rho^Z_i(z)=z_{\iota^Z(i)}$ for $i\in Z$, and $\pi^Z:X^\Box\to X^Z$ is the projection on the components of $Z$.\\
Then $f^Z$ is an autonomous network module, and is called the \emph{network module of $f$ derived from $Z$}.
\end{lem}
\begin{proof}
The function $f^Z$ maps regular states to regular states, since $Z$ is disjoint from $J(M)\setminus Z$ in $G^\theta(M)$ and thus $F_i(\rho^Z(z))\in X_i$ for $z\in X^Z$, $i\in Z$, according to Lemma~\ref{core} and the definition of $F$.
We now show $\pi^Z\circ f|_M=f^Z\circ\pi^Z|_M$. Let $j\in\haken{k}$ and $x\in M$. To simplify notation we drop the superscript $Z$ from $\pi^Z$, $\rho^Z$ and $\iota^Z$. We have
$F_{\iota(j)}(\rho(\pi(x)))=F_{\iota(j)}(x)=f_{\iota(j)}(x)$ according to Lemma~\ref{core} and the definition of $\rho$ and $F$. Thus,
$(\pi^Z\circ f|_M(x))_j=f_{\iota(j)}(x)=F_{\iota(j)}(x)=F_{\iota(j)}(\rho(\pi(x)))=(\pi\circ F\circ\rho(\pi(x)))_j=f^Z_j(\pi(x))$.
Furthermore, since $M$ is a symbolic steady state $X^Z=\pi(M)$ satisfies the conditions regarding the domain of an autonomous network module given in Def.~\ref{nwmoddef}.
\end{proof}
It is easy to see that the global interaction graph $G(f^Z)$ is isomorphic to $Z$, so the structural module derived from $M$ matches the one derived from $f^Z$.
The dynamical modules derived from all functions $f^Z$, $Z$ component of $G^\theta(M)$, i.\,e., the state transition graphs $S(f^Z)$, constitute a breakdown of the coarse dynamical module, i.\,e., the trap set $M$, which we used to determine first the structural and then the network modules. We end this section by showing that these finer dynamical modules are building blocks of the dynamics of $f$. More specific, we show that we can compose the state transition graph of $f$ from the state transition graphs of the network modules $f^Z$ and the frozen components of $M$. Again, in the following we generalize results from \cite{MCS}. First, we define the composition of the graphs $S(f^Z)$.
\begin{defn}\label{prod}
Let $M\in X^\Box$ be a symbolic steady state, and let $Z_1,\dots,Z_m$ be the components of $G^\theta(M)$. 
We then denote by $S^M$ the graph with vertex set $M$ and edge set $E^M$. An edge $x^1\to x^2$ belongs to the edge set iff
\[
x^1=x^2,\,\,\mbox{ and}\quad\,\pi^{Z_j}(x^1)\to\pi^{Z_j}(x^2)\quad\mbox{belongs to $\,S(f^{Z_j})\,$ for all }\,j\in\haken{m},\vspace{.5ex}
\]
 or if there exists $j\in\haken{m}$ such that
 \[
 \pi^{Z_j}(x^1)\to \pi^{Z_j}(x^2)\,\mbox{ is an edge in }\,S(f^{Z_j})\,\mbox{ and }\,x_i^1=x_i^2\,\mbox{ for all }\,i\notin V^{Z_j}.
\]
We call $S^M$ the product state transition graph corresponding to $M$.
\end{defn}
The next theorem confirms that the method of composing the state transition graphs of the network modules renders the subgraph of $S(f)$ derived from the state set $M$.
\begin{thm}\label{STGconstr}
Let $M\in X^\Box$ be a symbolic steady state, and let $Z_1,\dots,Z_m$ be the components of $G^\theta(M)$. Let $S(f)|_M$ denote the subgraph of $S(f)$ with vertex set $M$ and all edges in $S(f)$ with both end-vertices in $M$. Then $S(f)|_M=S^M$.
\end{thm}
\begin{proof}
Recall that $M$ is a trap set, and that for $x,x'\in M$ there is an edge $x\to x'$ in $S(f)$ if and only if $x'=f(x)=x$ or $x'_i=x_i+\sgn(f_i(x)-x_i)$ for some $i\in\haken{n}$ satisfying $x_i\neq f_i(x)$, and $x_j'=x_j$ for all $j\neq i$.

First, we note that $f_i(x)=x_i=M_i$ for all $x\in M$ and $i\notin J(M)$. Furthermore, for $i\in J(M)$ there exists $k^i\in\haken{m}$ such that $i\in V^{Z_{k^i}}$, and there exists $l^i$ such that $\iota^{Z_{k^i}}(l^i)=i$, with $\iota^{Z_{k^i}}$ being the bijection introduced in Def.~\ref{nwmod}. As seen in the proof of Lemma~\ref{aut}, we have $f_i(x)=f^{Z_{k^i}}_{l^i}(\pi^{Z_{k^i}}(x))$.
Therefore, $x\in M$ is a fixed point of $f$ iff $\pi^{Z_j}(x)$ is a fixed point of $f^{Z_j}$ for all $j\in\haken{m}$, and $x_i+\sgn(f_i(x)-x_i)=(\pi^{Z_{k^i}}(x))_{l^i}+\sgn(f^{Z_{k^i}}_{l^i}(\pi^{Z_{k^i}}(x))-(\pi^{Z_{k^i}}(x))_{l^i})$. Thus, we can construct each edge in $S(f)|_M$ from an edge in some $S(f^{Z_j})$ and vice versa. It follows from Def.~\ref{prod} that $S(f)|_M=S^M$.
\end{proof}
The reasoning in the proof of Theorem~\ref{STGconstr} leads immediately to the following statement.
\begin{cor}\label{construction}
Let $M$ and $Z_1,\dots, Z_m$ be as in Theorem~\ref{STGconstr}. For all $i\in\haken{m}$ let $A_i$ be an attractor in $S(f^{Z_i})$. Then $A:=\{a\in M\mid \forall j\in\haken{m}\,:\,\pi^{Z_j}(a)\in A_j\}$ is an attractor in $S(f)$. Moreover, every attractor in $S(f)|_M$ can be represented in this manner as product of attractors in $S(f_{Z_j})$, $j\in\haken{m}$, and component values $M_i$ for $i\notin J(M)$.
\end{cor}
We illustrate the results on our running example from Fig.~\ref{Ex}. The graph $G^\theta(M)$ for $M=([1,2],1,[0,1])$ has two components, $Z_1$ consisting of a positive loop on $\alpha_1$, and $Z_2$ being a negative loop on $\alpha_3$, as can be seen in Fig.~\ref{symb} (b). We define $f^{Z_1}:[1,2]\to[1,2]$ as in Def.~\ref{nwmod} with $\pi^{Z_1}((M'_1,M'_2,M'_3))=M'_1$ for all $M'\subseteq M$ and $\rho^{Z_1}(z)=(z,1,[0,1])$ for all $z\in[1,2]$. This reduces the function $f_1$ given in Fig.~\ref{Ex} (a) to $f^{Z_1}(z)=z$. The state transition graph $S(f^{Z_1})$ is given in Fig.~\ref{symb} (c) and consists of two steady states.  Analogously, we obtain $S(f^{Z_2})$, which comprises a single attractor of cardinality 2 as shown in Fig.~\ref{symb} (d). Applying Def.~\ref{prod} we obtain the graph $S^M$ which is shown in Fig.~\ref{symb} (e). Comparison with the state transition graph $S(f)$ given in Fig.~\ref{Ex} (c) illustrates Theorem~\ref{STGconstr} and its corollary.
\section{Networks with input layer}
\begin{figure}[tb]
\begin{center}
\includegraphics[width=.9\textwidth]{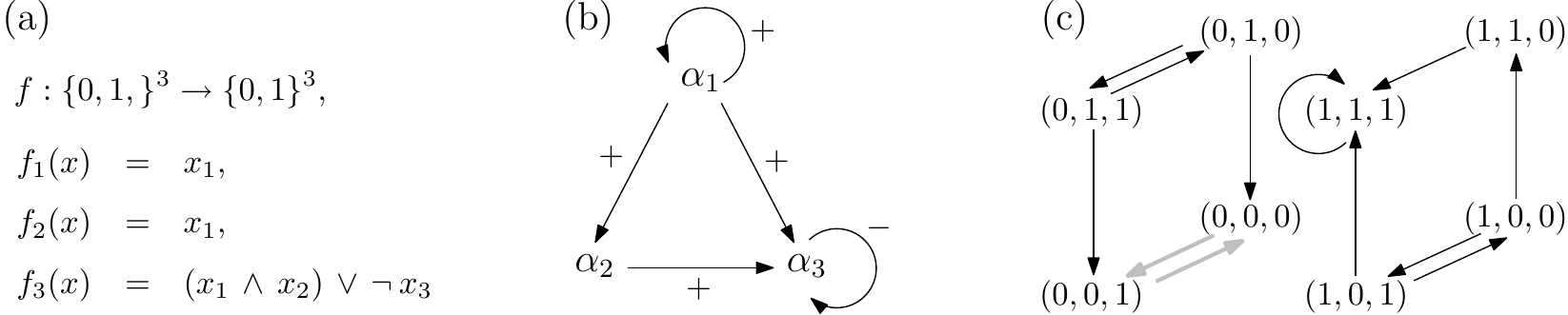}
\caption{\label{input} A network with input layer in (a), its global interaction graph in (b), and its state transition graph in (c). }
\end{center}
\end{figure}
In this section we introduce a class of networks, for which we can easily find a set of symbolic steady states such that all attractors of the original system can be constructed from the attractors of the network modules derived from the symbolic steady states. We extend results obtained in \cite{MCS}.
\begin{defn}
We call $f$ a \emph{network with input layer}, if there exists $i\in\haken{n}$ such that $f_i={\rm id}_{X_i}$. A vertex satisfying this condition is called \emph{input vertex}.
\end{defn}
For every input vertex $\alpha_i$, we have no incoming edges except a positive loop in the global interaction graph. However, this structural criterion is not sufficient for identifying an input vertex. Networks with input layer are well-suited for modeling, for example, signal transduction networks, where receptors may be modeled as input vertices.

In the following we assume that $f$ is a network with input layer. Without loss of generality we assume that $\alpha_1,\dots,\alpha_k$, $k\in\haken{n}$, are the input vertices of $f$.

\begin{figure}[tb]
\begin{center}
\includegraphics[width=.65\textwidth]{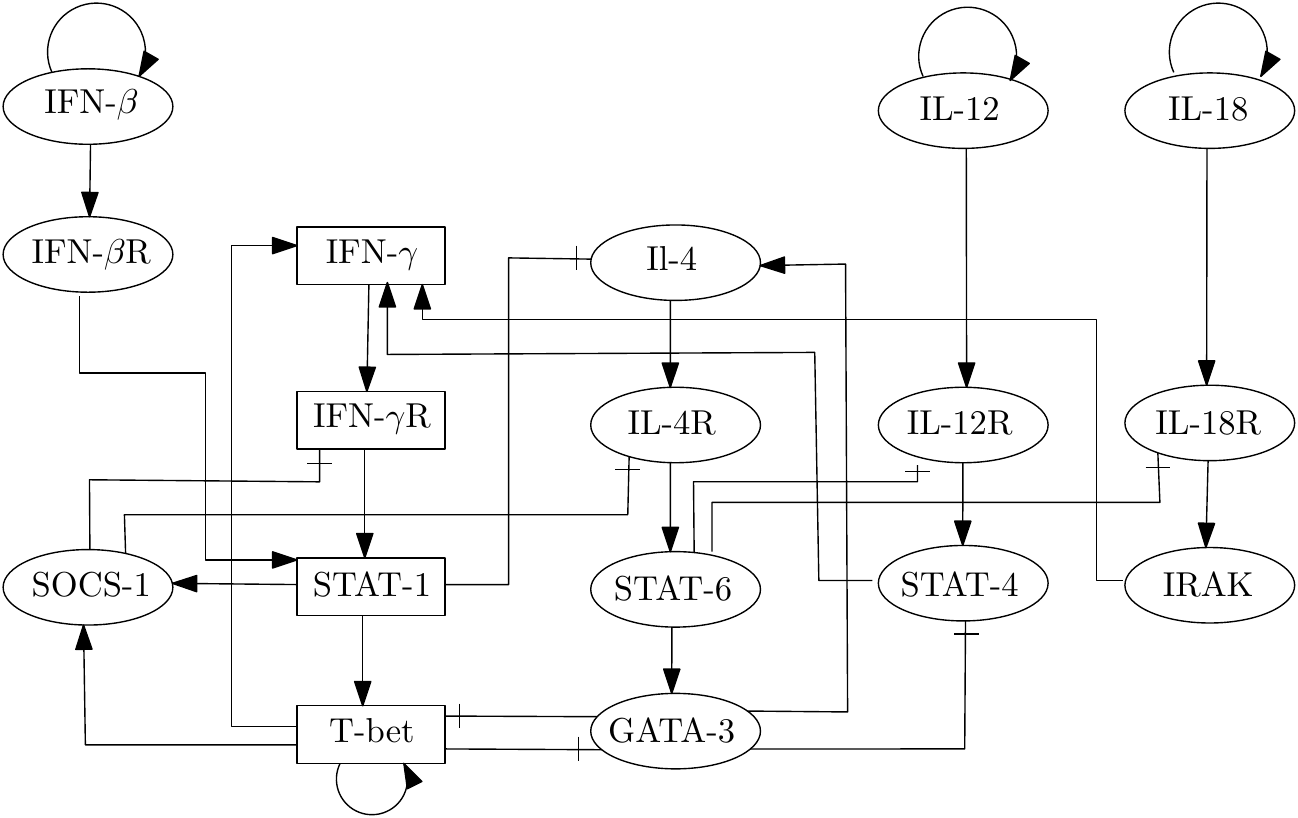}
\caption{\label{Tcell} Global interaction graph of the Th cell differentiation network introduced in \cite{Mendoza}. Arrows represent activation, crossed lines represent inhibition. }
\end{center}
\end{figure}
We can immediately note one important property of $f$. The frozen core of a symbolic state $M'\in X^\Box$ with $M'_i\in X_i$ for $i\in\haken{k}$ and $M'_i=X_i$ for $i>k$ is given by $(\haken{k},M')$, since $M'$ is its own extended orbit. In other words, we can easily derive a (symbolic) steady state from each combination of values for the input vertices using Theorem~\ref{frozen}. The resulting set of symbolic and possibly regular steady states is sufficient to determine all attractors of $f$.
\begin{thm}
Let $A$ be an attractor of $f$. Then there exist input values $x_i\in X_i$, $i\in\haken{k}$ such that either $A=M$ is a regular steady state, or we can construct $A$ from $M$ as shown in Cor~\ref{construction}, where $M$ is the fixed point derived from $(\haken{k},M')$, $M'\in X^\Box$ with $M'_i=x_i$ for $i\in\haken{k}$ and $M'_i=X_i$ for $i>k$.
\end{thm}
\begin{proof}
Since $\alpha_1\dots\alpha_k$ are input vertices, we have $a_i=a'_i$ for all $i\in\haken{k}$, and we set $M_i'=a_i$ for the input vertices. Next, we show that for every state $x\in X$ with $x_i=M_i'$ for all $i\in\haken{k}$, there exists a trajectory leading to $M$.

Let $(M^l)_{l\in\mathbb{N}_0}$ with $M^0:=\widetilde{M'}$ and $M^l:=F(M^{l-1})$ for all $l\in\mathbb{N}$ be the sequence converging to $M$ introduced in Theorem~\ref{frozen}. Recall that $M^l\subseteq M^{l-1}$ for all $l\in\mathbb{N}$. Let $x^0\in M'\setminus M$. Then there exists $l\in\mathbb{N}_0$ such that $x^0\in M^l\setminus M^{l+1}$. Then, according to the definition of $F$ and since $(M^l)_{l\in\mathbb{N}_0}$ is decreasing with respect to the subset relation, either all successors of $x^0$ in $S(f)$ are in $M^l\setminus M^{l+1}$ or there exists a successor in $M^{l+1}$. In the latter case, we label that successor $x^1$. Otherwise, we repeat the procedure for all successor of $x^0$, check again and if necessary repeat again. Since all images of all states in $M^l$ lie in $M^{l+1}$ and since the state space is finite, we eventually find  a state $x^1\in M^{l+1}$ such that there exists a path from $x^0$ to $x^1$ with all states of the path except $x^1$ lying in $M^l\setminus M^{l+1}$. Since $(M^l)_{l\in\mathbb{N}_0}$ is converging to $M$ we can thus construct a path from each state in $M'$ to $M$.

It follows that there is no trap set, and thus no attractor, in $M'\setminus M$. If $M$ is a regular steady state, then $M=A$ is the only attractor in $M'$. Otherwise $A$ is a composition of attractors of network modules derived from $M$ and the frozen components of $M$ as shown in Cor~\ref{construction}.
\end{proof}
We illustrate the results on the simple Boolean network given in Fig.~\ref{input} which has one input vertex, namely $\alpha_1$. We start calculating the two (symbolic) steady states from the input values $x_1=0$ and $x_1=1$. In the first case, we get $F((0,[0,1],[0,1]))=(0,0,[0,1])$ and $F((0,0,[0,1]))=(0,0,[0,1])$, i.\,e., $(0,0,[0,1])$ is a symbolic steady state, and the associated structural module is a negative loop originating in $\alpha_3$. When looking at the corresponding network module, we obtain the function $f^Z:[0,1]\to[0,1]$, $f^Z(z)=\neg\,z$, and the derived attractor for the original network is the set $\{(0,0,0),(0,0,1)\}$. For the input value $x_1=1$ we get the sequence $((1,[0,1],[0,1]),(1,1,[0,1]),(1,1,1), (1,1,1),\dots)$, so the procedure renders a regular steady state of the system. As can be seen in Fig.~\ref{input}, the regular steady state and the attractor derived from $(0,0,[0,1])$ are the only attractors of the system.
\section{Analyzing Th cell differentiation}\label{Tcellsect}
T helper cells, short Th cells, are important players in the vertebrate immune system. They can be sub-classified in Th1 and Th2 cells, which are involved in different immune responses. Both originate from a common precursor, promote their own differentiation and inhibit proliferation of each other. In \cite{Mendoza} L. Mendoza proposes a model for a control network of Th cell differentiation consisting of 17 components, 13 of which are represented by Boolean variables while the remaining 4 components have three activity levels. The logical rules governing the behavior of the system are given in Table~\ref{Tcelltable} and the global interaction graph can be seen in Fig.~\ref{Tcell}. Note that the model depicted in Fig.~\ref{Tcell} and Table~\ref{Tcelltable} differs slightly from the model introduced in \cite{Mendoza}, namely we altered the logical functions associated with the vertices  IFN-$\beta$, IL-12 and Il-18. In Mendoza's model, the three components are modeled with constant functions with value 0, representing the wild type in some sense (see \cite{Mendoza}). The constant values are changed when considering specific artificial environmental conditions. We model these vertices as input vertices and consider for the wild type the situation where all input values are zero. Clearly, the attractors of both models coincide. Modeling IFN-$\beta$, IL-12 and IL-18 as input vertices also makes sense from a biological point of view since all three vertices represent substances not reproduced by Th cells. If we want to mimic experimental conditions where cells are cultured in media saturated with one or more of these substances, we can easily do so by focussing on the part of state space where one or more of the input vertices have value one. Note that we have no further input vertices in our model.
\begin{table}[t]
\begin{center}{\small
\begin{tabular}{lll}
IFN-$\beta$ & $X_1=\{0,1\}$ &  $f_1(x)=x_1$ \\[1ex]
IL-12 & $X_2=\{0,1\}$ &  $f_2(x)=x_2$ \\[1ex]
IL-18 & $X_3=\{0,1\}$ &  $f_3(x)=x_3$ \\[1ex]
IFN-$\beta$R & $X_4=\{0,1\}$ &  $f_4(x)=x_1$ \\[1ex]
IFN-$\gamma$ & $X_5=\{0,1,2\}$ &  $f_5(x)=1$ if $(x_{16}=1\,\wedge\,\neg\,(x_{14}=1\,\wedge\, x_{15}=1))\,\vee\,(x_{14}=1\,\wedge\,x_{15}=x_{16}=0)$, \\[.5ex]
& & $f_5(x)=2$ if $x_{16}=2\,\vee\,(x_{14}=1\,\wedge\, x_{15}=1)$, and $f_5(x)=0$ otherwise\\[1ex]
IL-4R & $X_6=\{0,1\}$ &  $f_6(x)=1$ if $x_{12}=0\,\wedge\, x_{17}=1$, and $f_6(x)=0$ otherwise\\[1ex]
IFN-$\gamma$R & $X_7=\{0,1,2\}$ &  $f_7(x)=1$ if $x_5=1\,\vee\,(x_5=2\,\wedge\, x_{11}=1)$, \\[.5ex]
& & $f_7(x)=2$ if $x_5=2\,\wedge\, x_{11}=0$, and $f_7(x)=0$ otherwise\\[1ex]
IL-4R & $X_8=\{0,1\}$ &  $f_8(x)=x_6\,\wedge\,\neg\, x_{11}$ \\[1ex]
IL-12R & $X_9=\{0,1\}$ &  $f_9(x)=x_2\,\wedge\,\neg\, x_{13}$ \\[1ex]
IL-18R & $X_{10}=\{0,1\}$ &  $f_{10}(x)=x_3\,\wedge\,\neg\, x_{13}$ \\[1ex]
SOCS-1 & $X_{11}=\{0,1\}$ &  $f_{11}(x)=1$ if $x_{12}\geq 1\,\vee\, x_{16}\geq 1$, and $f_{11}(x)=0$ otherwise \\[1ex]
STAT-1 & $X_{12}=\{0,1,2\}$ &  $f_{12}(x)=1$ if $(x_4=1\,\wedge\, x_7=0)\,\vee\,x_7=1$, \\[.5ex]
& & $f_{12}(x)=2$ if $x_7=2$, and $f_{12}(x)=0$ otherwise\\[1ex]
STAT-6 & $X_{13}=\{0,1\}$ &  $f_{13}(x)=x_8$ \\[1ex]
STAT-4 & $X_{14}=\{0,1\}$ &  $f_4(x)=x_9\,\wedge\,\neg\,x_{17}$ \\[1ex]
IRAK & $X_{15}=\{0,1\}$ &  $f_{15}(x)=x_{10}$ \\[1ex]
T-bet & $X_{16}=\{0,1,2\}$ &  $f_{16}(x)=1$ if $(x_{17}=0\,\wedge\,((x_{12}=1\,\wedge\,x_{16}\leq 1)\,\vee\,(x_{12}\leq 1\,\wedge\,x_{16}=1)))$\\[.5ex]
& & \hspace*{5.5em}$\vee\,(x_{17}=1\,\wedge\,x_{16}=1\,\wedge\,x_{12}=1)$, \\[.5ex]
& & $f_{16}=2$ if $(x_{17}=0\,\wedge\,(x_{12}=2\,\vee\,x_{16}=2))\,\vee\,(x_{17}=1\,\wedge\,x_{12}=1\,\wedge\,x_{16}=2)$,\\[.5ex]
& & $f_{16}(x)=0$ otherwise\\[1ex]
GATA-3 & $X_{17}=\{0,1\}$ &  $f_{17}(x)=1$ if $x_{13}=1\,\wedge\,x_{16}=0$, and $f_{17}(x)=0$ otherwise\\[1ex]
\end{tabular}}
\caption{\label{Tcelltable} Coordinate functions and ranges for the components of the Th cell network.}
\end{center}
\end{table}

For the wild type, i.\,e. the situation where all input values are set to zero, Mendoza identifies four attractors all of which are fixed points of the function $f$ given in Table~\ref{Tcelltable}. Each one has a clear biological interpretation \cite{Mendoza}. We now want to apply our analysis technique using symbolic steady states to the wild type.

We fix the values of the input vertices to zero and as a first step determine the corresponding symbolic steady state, that is, the symbolic steady state derived from the frozen component set $(\{1,2,3\},\, x_1=x_2=x_3=0)$. Iterating the state $M^0:=(0,0,0,[0,1],[0,2],[0,1],[0,2],[0,1],[0,1],[0,1],[0,1],[0,2],[0,1],[0,1],$ $[0,1],[0,2],[0,1])$ we get\vspace{1.5ex}\newline
\begin{tabular}{lll}
$M^1:= $&\hspace*{-1.5em}$f(M^0)=$&\hspace*{-.5em}$(0,0,0,0,[0,2],[0,1],[0,2],[0,1],0,0,[0,1],[0,2],[0,1],[0,1],[0,1],[0,2],[0,1])$,\\[1ex]
$M^2:= $&\hspace*{-1.5em}$f(M^1)=$&\hspace*{-.5em}$(0,0,0,0,[0,2],[0,1],[0,2],[0,1],0,0,[0,1],[0,2],[0,1],0,0,[0,2],[0,1])$,\\[1ex]
$f(M^2)=$&\hspace*{-.5em}$M^2$.&
\end{tabular}\vspace{1.5ex}\newline
We obtain a symbolic steady state with 8 regular components, and no further constraints on the remaining components. The local interaction graph $G^\theta(M)$ is shown in Fig.~\ref{Tcellsub} (a). Analysis of the corresponding subnetwork renders four fixed points, namely $(0,0,0,0,0,0,0,0,0)$, $(1,0,1,0,1,1,0,1,0)$, $(2,0,1,0,1,1,0,2,0)$, $(0,1,0,1,0,0,1,0,1)\in X_5\times X_6\times X_7\times X_8\times X_{11}\times X_{12}\times X_{13}\times X_{16}\times X_{17}$. The steady states of the original network derived from these fixed points match the four steady states found in \cite{Mendoza}.
The state space of the original model consists of 663552 states. Fixing the input values still leaves us with 82944 states to consider. The state space of the structural module associated with the symbolic steady state $M^2$ contains only 2592 states.

Not all of the combinations of input values render a significant simplification of the network analysis. In the worst case, for example if we choose input values $x_1=0$, $x_2=x_3=1$, we can only derive the value for $x_4$ but no further constraints on structure and behavior of the system. This in itself is of course an interesting observation from a biological point of view, since in that case we can deduce that cross-regulation plays an important role at an early stage of signal transduction.

On the other hand, some combinations of input values lead to very small network modules. Let us as a last example consider the input values $x=1$, $x_2=x_3=0$, representing an overabundance of IFN-$\beta$. Starting with  $M^0:=(1,0,0,[0,1],[0,2],[0,1],[0,2],[0,1],[0,1],[0,1],[0,1],[0,2],[0,1],[0,1],[0,1],$ $[0,2],[0,1])$ we get\vspace{1.5ex}\newline
\begin{tabular}{lll}
$M^1:=$&\hspace*{-.5em}$f(M^0)=$&\hspace*{-.5em}$(1,0,0,1,[0,2],[0,1],[0,2],[0,1],0,0,[0,1],[0,2],[0,1],[0,1],[0,1],[0,2],[0,1])$,\\[1ex]
$M^2:=$&\hspace*{-.5em}$f(M^1)=$&\hspace*{-.5em}$(1,0,0,1,[0,2],[0,1],[0,2],[0,1],0,0,[0,1],[1,2],[0,1],0,0,[0,2],[0,1])$,\\[1ex]
$M^3:=$&\hspace*{-.5em}$f(M^2)=$&\hspace*{-.5em}$(1,0,0,1,[0,2],0,[0,2],[0,1],0,0,1,[1,2],[0,1],0,0,[0,2],[0,1])$,\\[1ex]
$M^4:=$&\hspace*{-.5em}$f(M^3)=$&\hspace*{-.5em}$(1,0,0,1,[0,2],0,[0,1],0,0,0,1,[1,2],[0,1],0,0,[0,2],[0,1])$,\\[1ex]
$M^5:=$&\hspace*{-.5em}$f(M^4)=$&\hspace*{-.5em}$(1,0,0,1,[0,2],0,[0,1],0,0,0,1,1,0,0,0,[0,2],[0,1])$,\\[1ex]
$M^6:=$&\hspace*{-.5em}$f(M^5)=$&\hspace*{-.5em}$(1,0,0,1,[0,2],0,[0,1],0,0,0,1,1,0,0,0,[0,2],0)$,\\[1ex]
$M^7:=$&\hspace*{-.5em}$f(M^6)=$&\hspace*{-.5em}$(1,0,0,1,[0,2],0,[0,1],0,0,0,1,1,0,0,0,[1,2],0)$,\\[1ex]
$M^8:=$&\hspace*{-.5em}$f(M^7)=$&\hspace*{-.5em}$(1,0,0,1,[1,2],0,[0,1],0,0,0,1,1,0,0,0,[1,2],0)$,\\[1ex]
$M^9:=$&\hspace*{-.5em}$f(M^8)=$&\hspace*{-.5em}$(1,0,0,1,[1,2],0,1,0,0,0,1,1,0,0,0,[1,2],0)$,\, $f(M^9)=M^9$.
\end{tabular}\vspace{1.5ex}\newline
We obtain a two-component module consisting of a positive loop on T-bet and an activating edge from T-bet to IFN-$\gamma$ which can be seen in Fig.~\ref{Tcellsub} (b). Both components of the module originally have three activity levels, but are both constrained to levels 1 and 2 in the module dynamics. Thus, we only have to analyze a state space of cardinality four instead of a state space consisting of 82944 states. The module has two steady states, namely $(x_5,x_{16})=(1,1)$ and $(x_5,x_{16})=(2,2)$, which translate to two steady states in the original network. Again, this is in agreement with the results in \cite{Mendoza} (supplementary material).

Application of our analysis method to this model thus offers two advantages. First, the complexity of the analysis of the dynamics is reduced, since we only have to focus on the smaller network modules. Furthermore, identification of the modules themselves is of interest, since they represent the part of the system responsible for the decision of the system's fate, i.\,e., which attractor is reached. An interesting next step would then be to check whether the mathematically derived network modules coincide with subsystems of known biological importance.
\begin{figure}[tb]
\begin{center}
\includegraphics[width=.65\textwidth]{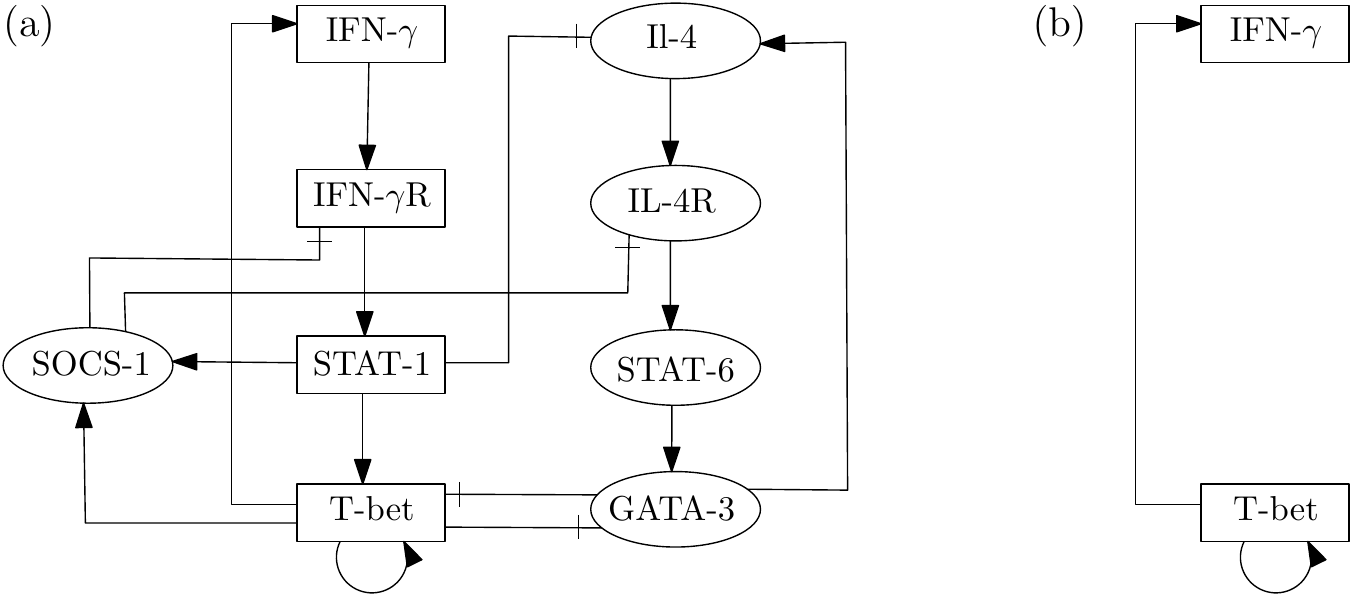}
\caption{\label{Tcellsub} Subnetworks of the Th cell differentiation network associated with the symbolic fixed points derived from the input values $x_1=x_2=x_3=0$ in (a) and $x_1=1$, $x_2=x_3=0$ in (b).}
\end{center}
\end{figure}
\section{Conclusion}
In this paper we introduced a method to analyze discrete regulatory networks utilizing suitable network modules. We used the notion of symbolic steady state, which allows us to determine such network modules using coarse dynamical and subsequently structural modules derived from the symbolic steady state. Lastly, we can associate network modules with the structural modules exploiting the properties of the frozen core of the symbolic steady state. We then can construct the dynamics of the original network, and in particular its attractors, in a subset of state space explicitly from the state transition graphs of the network modules. This paper not only gives a rigorous definition of different aspects of modularity but notably extends results in \cite{MCS,Siebert09-AB}. In particular, the detailed analysis of the Th cell network becomes possible because of the refined notion of symbolic steady state.

A variety of aspects provide possibilities for fruitful future work. Firstly, we want to focus on further options for easily computing symbolic steady states. Here, we introduced a method suited for networks with input-layer, but the resulting symbolic steady states might not be minimal with respect to the subset relation on symbolic steady states, and thus the resulting network modules can possibly be further refined. We anticipate results when we focus on certain classes of functions $f$ describing networks, in particular (nested) canalyzing functions \cite{Kauffman,Laubenbacher07}. A different direction of interest is to use the network modules not only for obtaining the system's dynamics but also for a refined stability analysis. Perturbations resulting in changes in the dynamics of $f$ might not be noticeable in every network module (see \cite{Irons}). Such considerations are of similar interest when using the synchronous update strategy, which is also often utilized in discrete modeling, so a translation of our results to synchronous update networks seems worthwhile. Furthermore, we need to compare our results to other well-established modularization techniques that aim at reducing the analysis complexity such as the modular response analysis of biochemical networks \cite{Kholodenko, Bruggeman}. Although the underlying modeling approaches are very different, some ideas may be transferable or render complementing results.
Lastly, we plan to apply our methods to further biological examples. Here, a comparison of the network modules and their associated structural and dynamical modules with subsystems of known biological importance could lead to interesting insights.

\bibliographystyle{eptcs} 
\bibliography{CompLit}
\end{document}